\documentclass[10pt]{iopart}

\usepackage{iopams}
\usepackage{amsthm}
\newtheorem{thm}{Theorem}
\newtheorem{lem}[thm]{Lemma}
\theoremstyle{definition}
\usepackage[varg]{txfonts}
\usepackage{bm}
\usepackage{subeqn}
\usepackage{enumerate}
\usepackage{color}
\usepackage{graphicx}
\begin{document}

\title[Combinatorial solutions to the discrete and ultradiscrete Toda molecules]{Combinatorial expressions of the solutions to initial value problems of the discrete and ultradiscrete Toda molecules}

\author{Shuhei Kamioka and Tomoaki Takagaki}

\address{Department of Applied Mathematics and Physics, Graduate School of Informatics, Kyoto University,
  Kyoto, Japan, 606-8501}
\begin{abstract}
  Combinatorial expressions are presented
  to the solutions to initial value problems of the discrete and ultradiscrete Toda molecules.
  For the discrete Toda molecule, a subtraction-free expression of the solution is derived
  in terms of non-intersecting paths,
  for which two results in combinatorics,
  Flajolet's interpretation of continued fractions and Gessel--Viennot's lemma on determinants, are applied.
  By ultradiscretizing the subtraction-free expression,
  the solution to the ultradiscrete Toda molecule is obtained.
  It is finally shown that the initial value problem of the ultradiscrete Toda molecule is exactly solved
  in terms of shortest paths on a specific graph.
\end{abstract}


\section{Introduction}

The Toda molecule \cite{Hirota(1988TM)} is a semi-infinite version of the Toda lattice \cite{Toda(1967)}.
As is the Toda lattice, the Toda molecule is known as a typical example of integrable systems
which has, despite its nonlinearity, an exact solution in terms of Hankel determinants.
The {\em discrete Toda molecule} is a discrete analogue of the Toda molecule
which is derived in \cite{Hirota-Tsujimoto-Imai(1993RIMS)} by using the bilinear formalism.
The discrete Toda molecule is a discrete integrable system which possesses a Hankel determinant solution
analogous to the Toda molecule.

The discrete Toda molecule is also derived by using the Lax formalism
\cite{Papageorgiou-Grammaticos-Ramani(1995),Maeda-Tsujimoto(2013pre)},
in which a connection with orthogonal polynomials is exploited to deduce the time evolution equations of
the discrete Toda molecule
\begin{subequations} \label{eq:dToda}
  \begin{eqnarray}
    q^{(t+1)}_{n} + e^{(t+1)}_{n} = q^{(t)}_{n} + e^{(t)}_{n+1}, \\
    q^{(t+1)}_{n} e^{(t+1)}_{n+1} = q^{(t)}_{n+1} e^{(t)}_{n+1}, \\
    e^{(t)}_{0} = 0 \qquad  \mbox{for $t \in \mathbb{Z}$, $n \in \mathbb{N}_{0}$.}
  \end{eqnarray}
\end{subequations}
In numerical algorithms, the equations \eref{eq:dToda} are known as recurrence equations of
the {\em qd algorithm}, which is used for computing Pad\'e approximants of analytic functions
(see, e.g., \cite{Baker-GravesMorris(1996)}),
and for computing eigenvalues of tridiagonal matrices (see, e.g., \cite{Rutishauser(1990)}).
In the study of pure combinatorics, Viennot \cite{Viennot(2000)} applied the qd algorithm to
a combinatorial problem of enumerating configurations of non-intersecting paths.
From the viewpoint of dynamical systems, Viennot's combinatorial result is observed as solving
an initial value problem of the discrete Toda molecule \eref{eq:dToda} from particular initial value.

In recent progress of investigating integrable systems, much attention has been given to
ultradiscrete integrable systems,
especially since the discovery of a direct connection between
discrete integrable systems and soliton cellular automata \cite{Tokihiro-Takahashi-Matsukidaira-Satsuma(1996)}.
A general method to derive an ultradiscrete integrable system from a discrete integrable system is called
{\em ultradiscretization}.
For the discrete Toda molecule \eref{eq:dToda}, {\em ultradiscretization} is performed as follows
\cite{Nagai-Takahashi-Tokihiro(1999)}:
Introduce new dependent variables $Q^{(t)}_{n}$, $E^{(t)}_{n}$
by $q^{(t)}_{n} = \exp(-Q^{(t)}_{n}/\varepsilon)$, $e^{(t)}_{n} = \exp(-E^{(t)}_{n}/\varepsilon)$
with a parameter $\varepsilon > 0$, and take the limit $\varepsilon \to 0$.
The equations \eref{eq:dToda} then tend to the {\em ultradiscrete Toda molecule}
\begin{subequations} \label{eq:uToda}
  \begin{eqnarray}
    Q^{(t+1)}_{n} = {}
    \min{\left\{ \sum_{k=0}^{n} Q^{(t)}_{k} - \sum_{k=0}^{n-1} Q^{(t+1)}_{k}, E^{(t)}_{n+1} \right\}}, \\
    E^{(t+1)}_{n+1} = Q^{(t)}_{n+1} - Q^{(t+1)}_{n} + E^{(t)}_{n+1}, \qquad
    \mbox{for $t \in \mathbb{Z}$, $n \in \mathbb{N}_{0}$.}
  \end{eqnarray}
\end{subequations}
It is shown in \cite{Nagai-Takahashi-Tokihiro(1999)} that
the ultradiscrete Toda molecule \eref{eq:uToda} describes the dynamics of a box-ball system.

In this paper, we examine the initial value problems of the discrete and ultradiscrete Toda molecules,
\eref{eq:dToda} and \eref{eq:uToda}, for which the initial value is given at $t=0$.
Obviously, one can exactly solve the problems in the following sense:
At any time $t \in \mathbb{N}_{0}$ and any site $n \in \mathbb{N}_{0}$,
the exact value of each dependent variable can be calculated from the initial value
in finitely many arithmetic and minimizing operations.
However, it is still nontrivial how to formulate the solutions since the equations are nonlinear.
The aim of this paper is to derive an exact expression of the solutions to the initial value problems
purely in terms of the initial value.
In order to formulate the solutions, we will utilize combinatorial objects,
non-intersecting paths and shortest paths on a graph,
in view of the combinatorial results on paths:
Flajolet's interpretation of continued fractions \cite{Flajolet(1980)} and
Gessel--Viennot's lemma \cite{Gessel-Viennot(1985),Aigner(2001LNCS)} on determinants.

This paper is organized as follows.
In \sref{sec:yudhisye}, we review a determinant solution to the discrete Toda molecule,
based on which, in \sref{sec:e6ucbvdfu}, we combinatorially formulate an exact expression of
the solution to the initial value problem of the discrete Toda molecule in terms of non-intersecting paths.
In \sref{sec:hdguiwefh7}, we derive the solution to the initial value problem of
the ultradiscrete Toda molecule by ultradiscretizing the solution to the discrete Toda molecule
obtained in \sref{sec:e6ucbvdfu}.
Further combinatorial observations lead us to a simpler expression of the solution in terms of shortest paths
on a specific graph.
\Sref{sec:ckdoc89p} is devoted to concluding remarks.

\section{Determinant solution to the discrete Toda molecule}
\label{sec:yudhisye}

In \sref{sec:yudhisye}, we give a brief review on a determinant solution to the discrete Toda lattice
together with bilinear equations associated with the discrete and ultradiscrete Toda molecules.
See, e.g., \cite{NakamuraY(1998),Nagai-Takahashi-Tokihiro(1999)} for detailed explanations.
Based on the determinant solution, in the subsequent sections,
we will examine initial value problems of the discrete and ultradiscrete Toda molecules.

We introduce a tau function $\tau^{(t)}_{n}$ of the discrete Toda molecule \eref{eq:dToda}
by the variable transformation
\begin{equation} \label{eq:hbgydjsuck}
  q^{(t)}_{n}   = \frac{\tau^{(t+1)}_{n+1} \tau^{(t)}_{n}}{\tau^{(t+1)}_{n} \tau^{(t)}_{n+1}}, \qquad
  e^{(t)}_{n+1} = \frac{\tau^{(t+1)}_{n} \tau^{(t)}_{n+2}}{\tau^{(t+1)}_{n+1} \tau^{(t)}_{n+1}}, \qquad
  n \in \mathbb{N}_{0},
\end{equation}
for which we assume the boundary condition that $\tau^{(t)}_{0} = 1$.
We then obtain from \eref{eq:dToda} a bilinear equation of the discrete Toda molecule,
\begin{equation} \label{eq:suhifefvs}
  \tau^{(t+1)}_{n+1} \tau^{(t-1)}_{n+1} = \tau^{(t+1)}_{n} \tau^{(t-1)}_{n+2} + \tau^{(t)}_{n+1} \tau^{(t)}_{n+1},
  \qquad n \in \mathbb{N}_{0}.
\end{equation}
\begin{subequations} \label{eq:hcudisgdv}
  To the bilinear equation \eref{eq:suhifefvs}, we have an exact solution in the Hankel determinant of size $n$
  \begin{equation} \label{eq:apo0g4wg}
    \tau^{(t)}_{n} = \det(f^{(t)}_{j+k})_{j,k=0}^{n-1} = {}
    {\left| \begin{array}{cccc}
        f^{(t)}_{0}   & f^{(t)}_{1} & \cdots & f^{(t)}_{n-1}  \\[1\jot]
        f^{(t)}_{1}   & f^{(t)}_{2} & \cdots & f^{(t)}_{n}    \\[1\jot]
        \vdots        & \vdots      & \ddots & \vdots         \\[1\jot]
        f^{(t)}_{n-1} & f^{(t)}_{n} & \cdots & f^{(t)}_{2n-2} \\[1\jot]
      \end{array} \right|}
  \end{equation}
  where $f^{(t)}_{n}$ is an arbitrary function subject to the linear {\em dispersion relation}
  \begin{equation} \label{eq:dsjuvdsvds}
    f^{(t+1)}_{n} = f^{(t)}_{n+1}, \qquad n \in \mathbb{N}_{0}.
  \end{equation}
  (The determinant of size zero is assume to be unity conventionally.)
\end{subequations}
We can verify the solution \eref{eq:hcudisgdv} by means of Sylvester's determinant identity.
Substituting the determinant \eref{eq:apo0g4wg} to \eref{eq:hbgydjsuck},
we obtain an exact solution to the discrete Toda molecule \eref{eq:dToda}.

We can derive a bilinear equation of the ultradiscrete Toda molecule \eref{eq:uToda}
by ultradiscretizing the procedure for the discrete Toda molecule:
Introducing a tau function $T^{(t)}_{n}$ by
\begin{subequations} \label{eq:oje8wxas}
  \begin{eqnarray}
    Q^{(t)}_{n} = T^{(t+1)}_{n+1} + T^{(t)}_{n} - T^{(t+1)}_{n} - T^{(t)}_{n+1}, \\
    E^{(t)}_{n+1} = T^{(t+1)}_{n} + T^{(t)}_{n+2} - T^{(t+1)}_{n+1} - T^{(t)}_{n+1}, \qquad
    n \in \mathbb{N}_{0},
  \end{eqnarray}
\end{subequations}
with $T^{(t)}_{0} = 0$,
we then obtain from \eref{eq:uToda} a bilinear equation of the ultradiscrete Toda molecule,
\begin{equation} \label{eq:ajekuygw}
  T^{(t+1)}_{n+1} + T^{(t-1)}_{n+1} = \min{\{ T^{(t+1)}_{n} + T^{(t-1)}_{n+2}, 2 T^{(t)}_{n+1} \}}, \qquad
  n \in \mathbb{N}_{0}.
\end{equation}

The equations \eref{eq:oje8wxas}, \eref{eq:ajekuygw} for the ultradiscrete Toda molecule are obtained
by ultradiscretizing the corresponding \eref{eq:hbgydjsuck}, \eref{eq:suhifefvs} for the discrete Toda molecule:
Assume that $q^{(t)}_{n} = \exp(-Q^{(t)}_{n}/\varepsilon)$, $e^{(t)}_{n} = \exp(-E^{(t)}_{n}/\varepsilon)$,
$\tau^{(t)}_{n} = \exp(-T^{(t)}_{n}/\varepsilon)$ with $\varepsilon > 0$ and take the limit $\varepsilon \to 0$.
Then, the equations \eref{eq:hbgydjsuck}, \eref{eq:suhifefvs} tend to \eref{eq:oje8wxas}, \eref{eq:ajekuygw},
respectively.
The limiting procedure in ultradiscretization replaces the operations $\times$, $/$, $+$ with $+$, $-$, $\min$,
respectively, for we have
\begin{subequations}
  \begin{eqnarray}
    -\varepsilon \log(\rme^{-X/\varepsilon} \times \rme^{-Y/\varepsilon}) = X+Y, \\
    -\varepsilon \log(\rme^{-X/\varepsilon} / \rme^{-Y/\varepsilon}) = X-Y, \\
    -\varepsilon \log(\rme^{-X/\varepsilon} + \rme^{-Y/\varepsilon}) \to \min{\{ X,Y \}} \qquad
    \mbox{as $\varepsilon \to 0$.}
  \end{eqnarray}
\end{subequations}
However, the counterpart of subtraction, $-$, does not exist since the limit of
$\varepsilon \log(\rme^{-X/\varepsilon}-\rme^{-Y/\varepsilon})$ as $\varepsilon \to 0$ is undetermined.
Commonly, we cannot ultradiscretize equations containing subtractions.
Especially, we cannot directly ultradiscretize the determinant solution \eref{eq:hcudisgdv}
because we may encounter subtractions in expanding the determinant.
In \sref{sec:e6ucbvdfu}, we derive a {\em subtraction-free} expression of the determinant solution
\eref{eq:hcudisgdv} to which ultradiscretization is directly applicable.

\section{Initial value problem of the discrete Toda molecule}
\label{sec:e6ucbvdfu}

As an initial value problem of the discrete Toda molecule, we consider the following:

{\em
  For the discrete Toda molecule (\ref{eq:dToda}), let us write the initial value at $t=0$
  \begin{equation} \label{eq:jbhuy7y8c}
    q^{(0)}_{n} = a_{2n}, \qquad
    e^{(0)}_{n+1} = a_{2n+1}, \qquad
    n \in \mathbb{N}_{0}.
  \end{equation}
  Then, for every $t \in \mathbb{N}_{0}$,
  find the exact value of $q^{(t)}_{n}$, $e^{(t)}_{n+1}$, $n \in \mathbb{N}_{0}$,
  uniquely determined from (\ref{eq:dToda}) in terms of the initial value $a_n$.}

To the qd algorithm for Pad\'e approximants whose recurrence equations are given by (\ref{eq:dToda}),
a combinatorial interpretation was given by Viennot \cite{Viennot(2000)},
in which a combinatorial expression of the determinant \eref{eq:apo0g4wg}
is formulated in terms of non-intersecting paths.
The fundamental idea used in \sref{sec:e6ucbvdfu} comes from Viennot's approach to the qd algorithm.

The discrete Toda molecule \eref{eq:dToda} is linearized in the following sense:
The nonlinear system \eref{eq:dToda} for $q^{(t)}_{n}$, $e^{(t)}_{n}$ reduces
into the linear system \eref{eq:dsjuvdsvds} for $f^{(t)}_{n}$
through the dependent variable transformations \eref{eq:hbgydjsuck} and \eref{eq:apo0g4wg}
via the tau function $\tau^{(t)}_{n}$.
We can thus evaluate the time evolution of the discrete Toda molecule by the dispersion relation \eref{eq:dsjuvdsvds}
whose initial value problem is exactly solved by
\begin{equation} \label{eq:hcuwebvd}
  f^{(t)}_{n} = f^{(0)}_{t+n}, \qquad t,n \in \mathbb{N}_{0}.
\end{equation}

The initial value problem therefore amounts to the following two subproblems:
\begin{enumerate}[(i)]
\item \label{itm:idadsvsd}
  Find the initial value $f^{(0)}_{n}$ of $f^{(t)}_{n}$ at $t=0$ in terms of $a_n$
  from \eref{eq:jbhuy7y8c}, \eref{eq:hbgydjsuck} and \eref{eq:hcudisgdv}.
\item \label{itm:abksuhv;}
  Find the value of the determinant \eref{eq:apo0g4wg} to evaluate the tau function $\tau^{(t)}_{n}$
  for each $t \in \mathbb{N}_{0}$.
\end{enumerate}

In order to solve the subproblem (\ref{itm:idadsvsd}),
we utilize Flajolet's combinatorial interpretation of continued fractions \cite{Flajolet(1980)}:
Let us consider a {\em path} $P$ in $\mathbb{Z}^{2}$
consisting of {\em up steps} $U = (1,1)$ and {\em down steps} $D = (1,-1)$.
We say that $P$ is {\em positive} if $P$ never goes beneath the $x$-axis, $y=0$.
(A positive path can touch the $x$-axis.)
We say that $P$ is {\em grounded} if both the two ends, the initial and terminal points, of $P$ are on the $x$-axis.
\Fref{fig:mvfhudsci} shows an example of a positive grounded path $P$.
\begin{figure}
  \centering
  \includegraphics{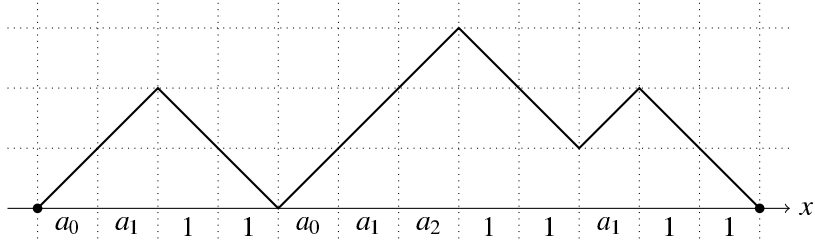}
  \caption{
    A positive grounded path $P$ of steps $UUDDUUUDDUDD$.
    The label of each step, $a_n$ or $1$, is shown below the step.
    The path $P$ weighs $w(P) = a_0^2 a_1^3 a_2$.
  }
  \label{fig:mvfhudsci}
\end{figure}
We label each step in $P$ by $a_n$ if the step is an up step ascending from the line $y=n$
and by unity if a down step.
We then define the weight $w(P)$ of $P$ by the product of the labels of all the steps in $P$.
For example, the path $P$ in \fref{fig:mvfhudsci} weighs $w(P) = a_0^2 a_1^3 a_2$.
We conventionally assume the weight of empty paths of no steps to be unity.
We refer by $D(P)$ to the number of down steps in $P$.

\begin{lem}[Flajolet \cite{Flajolet(1980)}] \label{lem:dtw67cuh}
  It holds that
  \begin{equation}
    \sum_{P} w(P) z^{D(P)} = {}
    \frac{1}{\displaystyle 1 - {}
      \frac{a_0 z}{\displaystyle 1 - {}
        \frac{a_1 z}{\displaystyle 1 - {}
          \frac{a_2 z}{\displaystyle 1 - \cdots}}}}
  \end{equation}
  where the (formal) sum in the left-hand side is taken over all the positive grounded paths $P$
  whose initial point is fixed at $(0,0)$.
\end{lem}

\begin{subequations} \label{eq:mcnohgj9}
  The subproblem (\ref{itm:idadsvsd}) asks us to solve the system of equations
  \begin{equation}
    a_{2n}   = \frac{\Delta'_{n+1} \Delta_{n}}{\Delta'_{n} \Delta_{n+1}}, \qquad
    a_{2n+1} = \frac{\Delta'_{n} \Delta_{n+2}}{\Delta'_{n+1} \Delta_{n+1}}, \qquad
    n \in \mathbb{N}_{0},
  \end{equation}
  for $f^{(0)}_{n}$, $n \in \mathbb{N}_{0}$, where $\Delta_n$ and $\Delta'_n$ denote the determinants of size $n$
  \begin{equation}
    \Delta_n = \det(f^{(0)}_{j+k})_{j,k=0}^{n-1}, \qquad
    \Delta'_n = \det(f^{(0)}_{j+k+1})_{j,k=0}^{n-1}.
  \end{equation}
\end{subequations}
In the theory of Pad\'e approximants (see, e.g., \cite{Baker-GravesMorris(1996)}),
it is well-known that the system \eref{eq:mcnohgj9} is equivalent
to the equation between a formal power series and an S-fraction
\begin{equation}
  \sum_{n=0}^{\infty} f^{(0)}_{n} z^{n} = {}
  \frac{1}{\displaystyle 1 - {}
    \frac{a_0 z}{\displaystyle 1 - {}
      \frac{a_1 z}{\displaystyle 1 - {}
        \frac{a_2 z}{\displaystyle 1 - \cdots}}}}
\end{equation}
with the normalization that $f^{(0)}_{0} = 1$.

This observation on Pad\'e approximants leads us to the solution to the subproblem (\ref{itm:idadsvsd}):
Owing to lemma \ref{lem:dtw67cuh}, with the normalization that $f^{(0)}_{0} = 1$,
\begin{equation} \label{eq:jtdsugfud}
  f^{(0)}_{n} = \sum_{P} w(P)
\end{equation}
where the sum in the right-hand side is taken over all the positive grounded paths $P$
whose two ends are fixed at $(0,0)$ and $(2n,0)$.
For example, the first few of $f^{(0)}_{n}$ are
\begin{subequations}
  \begin{eqnarray}
    f^{(0)}_{0} = 1, \\
    f^{(0)}_{1} = a_0, \\
    f^{(0)}_{2} = a_0^2 + a_0 a_1, \\
    f^{(0)}_{3} = a_0^3 + 2 a_0^2 a_1 + a_0 a_1^2 + a_0 a_1 a_2.
  \end{eqnarray}
\end{subequations}
As noted in \cite{Flajolet(1980)}, we can write the combinatorial formula \eref{eq:jtdsugfud} in the form
\begin{equation}
  f^{(0)}_{n} = {}
  \sum_{k_1=0} \sum_{k_2=0}^{k_1+1} \sum_{k_3=0}^{k_2+1} \cdots \sum_{k_n=0}^{k_{n-1}+1}
  a_{k_1} a_{k_2} a_{k_3} \cdots a_{k_n}.
\end{equation}
The initial value $f^{(0)}_{n}$ of $f^{(t)}_{n}$ is thus found as a polynomial in $a_k$ homogeneous of degree $n$.
The number of monomials in $f^{(0)}_{n}$,
which is equal to the number of positive grounded paths from $(0,0)$ to $(2n,0)$,
is counted by the Catalan number $C_n = \frac{1}{n+1} {2n \choose n}$.
For details on the Catalan numbers,
refer the On-Line Encyclopedia of Integer Sequences \cite{OEIS} for Sequence A000108.

Due to (\ref{eq:hcuwebvd}) and (\ref{eq:jtdsugfud}),
the $(j,k)$-entry of the determinant (\ref{eq:apo0g4wg}), $f^{(t)}_{j+k} = f^{(0)}_{t+j+k}$, is shown
to be equal to the sum of the weight $w(P)$ of all the positive grounded paths $P$
going from $(0,0)$ to $(2t+2j+2k,0)$, or equivalently, going from $(-2j,0)$ to $(2t+2k,0)$.
We can thereby successfully apply Gessel--Viennot's lemma on determinants
\cite{Gessel-Viennot(1985),Aigner(2001LNCS)}
to solve the subproblem (\ref{itm:abksuhv;}):
For $t,n \in \mathbb{N}_{0}$, let $P(t,n)$ denote the collection of $n$-sets $\bm{P} = {\{ P_0,\ldots,P_{n-1} \}}$
of positive grounded paths $P_j$ satisfying the following conditions:
\begin{enumerate}[(a)]
\item
  The two ends of $P_j$ are fixed at $(-2j,0)$ and $(2t+2j,0)$.
\item
  The $n$ paths $P_0,\ldots,P_{n-1}$ are {\em non-intersecting},
  Namely, every two distinct paths $P_j$ and $P_k$, $j \neq k$, never intersect at any points.
\end{enumerate}
Then, with the normalization that $f^{(0)}_{0} = 1$, Gessel--Viennot's lemma yields that
\begin{equation} \label{eq:6tcdscds}
  \tau^{(t)}_{n} = \sum_{\bm{P} \in P(t,n)} w(\bm{P}) \qquad
  \mbox{where $w(\bm{P}) = w(P_0) \cdots w(P_{n-1})$.}
\end{equation}
As shown in \fref{fig:iaufg8vyu},
we can draw each $n$-set $\bm{P} = {\{ P_0,\ldots,P_{n-1} \}} \in P(t,n)$ as a diagram
of $n$ positive grounded paths which are non-intersecting.
\begin{figure}
  \centering
  \includegraphics{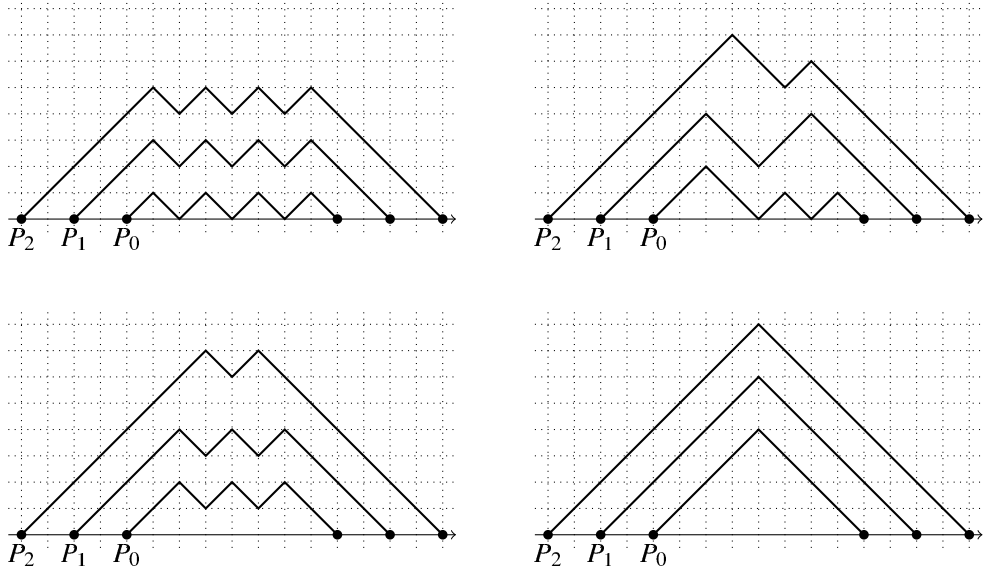}
  \caption{
    The $n$-sets $\bm{P} = {\{ P_0,\ldots,P_{n-1} \}} \in P(t,n)$ of positive grounded paths,
    where $t=4$, $n=3$.
    Each $n$-set $\bm{P}$ can be drawn in a diagram of $n$ non-intersecting positive grounded paths $P_j$
    such that $P_j$ goes from $(-2j,0)$ to $(2t+2j,0)$.
  }
  \label{fig:iaufg8vyu}
\end{figure}

The value of $\tau^{(t)}_{n}$ is found as a polynomial in $a_k$ homogeneous of degree $n(2t+n-1)/2$.
The number of monomials in $\tau^{(t)}_{n}$, which is equal to the cardinality $\#{P(t,n)}$, is exactly evaluated
in \cite{Viennot(2000)},
\begin{equation} \label{eq:fweo8ger}
  \#{P(t,n)} = \prod_{1 \le j \le k < t} \frac{2n+j+k}{j+k}.  
\end{equation}

We have solved the subproblems (\ref{itm:idadsvsd}) and (\ref{itm:abksuhv;}).
The solution to the initial value problem of the discrete Toda molecule is given as follows:

\begin{thm} \label{thm:dTodaTau}
  The solution to the initial value problem of the discrete Toda molecule \eref{eq:dToda} is given
  by (\ref{eq:hbgydjsuck}) with the tau function
  \begin{equation} \label{eq:iu8o7yew}
    \tau^{(t)}_{n} = \sum_{\bm{P} \in P(t,n)} w(\bm{P}), \qquad t,n \in \mathbb{N}_{0}.
  \end{equation}
\end{thm}

The expression \eref{eq:iu8o7yew} of the tau function $\tau^{(t)}_{n}$ is {\em subtraction-free},
namely, contains no subtractions.
That is why $\tau^{(t)}_{n}$ is positive for every $t,n \in \mathbb{N}_{0}$
if and only if the initial value $a_n$ is positive for every $n \in \mathbb{N}_{0}$.
The subtraction-free expression \eref{eq:iu8o7yew} of $\tau^{(t)}_{n}$ is {\em ultradiscretizable}
to obtain an exact solution to the ultradiscrete Toda molecule.

\section{Initial value problem of the ultradiscrete Toda molecule}
\label{sec:hdguiwefh7}

As an initial value problem of the ultradiscrete Toda molecule \eref{eq:uToda},
we consider the totally analogous problem to the discrete Toda molecule solved in \sref{sec:e6ucbvdfu}:

{\em
  For the ultradiscrete Toda molecule (\ref{eq:uToda}), let us write the initial value at $t=0$
  \begin{equation} 
    Q^{(0)}_{n} = A_{2n}, \qquad
    E^{(0)}_{n+1} = A_{2n+1}, \qquad
    n \in \mathbb{N}_{0}.
  \end{equation}
  Then, for every $t \in \mathbb{N}_{0}$,
  find the exact value of $Q^{(t)}_{n}$, $E^{(t)}_{n+1}$, $n \in \mathbb{N}_{0}$,
  uniquely determined from (\ref{eq:uToda}) in terms of the initial value $A_n$.}

The solution to this initial value problem can be obtained
by ultradiscretizing the corresponding solution to the discrete Toda molecule
stated in theorem \ref{thm:dTodaTau}.
Ultradiscretizing theorem \ref{thm:dTodaTau}, we obtain the following statement:

Let $P$ be a positive grounded path.
We label each step in $P$ by $A_n$ if the step is an up step ascending from the line $y=n$
and by {\em zero} if a down step.
We define the weight $W(P)$ by the {\em sum} of the labels of all the steps in $P$.
The solution to the initial value problem of the ultradiscrete Toda molecule is then given by (\ref{eq:oje8wxas})
with the tau function
\begin{equation} \label{eq:nsiygivf}
  T^{(t)}_{n} = \min_{\bm{P} \in P(t,n)} W(\bm{P}), \qquad W(\bm{P}) = W(P_0) + \cdots + W(P_{n-1}).
\end{equation}
The support $P(t,n)$ of the minimum is much the same as the sum in \eref{eq:iu8o7yew}.

The rest of \sref{sec:hdguiwefh7} is devoted
to simplifying the expression \eref{eq:nsiygivf} of the tau function $T^{(t)}_{n}$ by combinatorial observation.

\subsection{Solution in tabular paths}
\label{sec:kyifchid}

Let $P$ be a positive grounded path.
We refer to two consecutive up-down steps and down-up steps, $UD$ and $DU$, in $P$
by a {\em peak} and a {\em valley}, respectively.
We say that $P$ is {\em tabular} provided that, for some $k \in \mathbb{N}_{0}$,
all the peaks and the valleys in $P$ reside
in the strip of height one bordered by the two horizontal lines $y=k$ and $y=k+1$.

For $t,n \in \mathbb{N}_{0}$, we define a subset $\bar{P}(t,n) \subseteq P(t,n)$
as the collection of $n$-sets $\bar{\bm{P}} = {\{ \bar{P}_0,\ldots,\bar{P}_{n-1} \}} \in P(t,n)$
in which every $\bar{P}_j$ is tabular.
For example, all the $n$-sets $\bm{P} \in P(t,n)$ in \fref{fig:iaufg8vyu}, except the upper right one,
belong to $\bar{P}(t,n)$ since each path in $\bm{P}$ is tabular.

\begin{lem} \label{lem:7gidlvbl}
  There exists $\bar{\bm{P}} \in \bar{P}(t,n)$ which takes the minimal weight:
  $W(\bar{\bm{P}}) = \min_{\bm{P} \in P(t,n)} W(\bm{P})$.
\end{lem}

\begin{proof}
  We will prove the lemma by examining particular subpaths which we call {\em hooks}.
  Let $P$ be a positive grounded path.
  We call a subpath $H$ of $P$ an {\em up hook} (resp.~a {\em down hook})
  provided that $H$ is of at least four steps
  and that both the first and the last steps of $H$ are down steps (resp.~up steps)
  and all the middle steps are up steps (resp.~down steps).
  That is, each up hook (resp.~down hook) is of the form $D U^k D$ (resp.~$U D^k U$) for some integer $k \ge 2$.
  For example, see \fref{fig:hooks}.
  \begin{figure}
    \centering
    \includegraphics{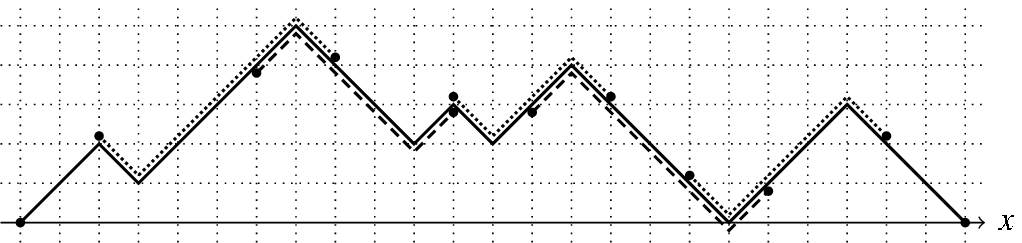}
    \caption{
      The hooks $H$ in a positive grounded path $P$.
      The path $P$ contains three up hooks, highlighted in dotted lines,
      and two down hooks, in dashed lines.
    }
    \label{fig:hooks}
  \end{figure}
  Obviously, a positive grounded path $P$ is tabular if and only if $P$ contains no hooks.

  Let us define two maps $\varphi$ and $\psi$ which deform a positive grounded path $P$ as follows:
  $\varphi(P)$ (resp.~$\psi(P)$) denotes the positive grounded path obtained from $P$
  by replacing each up hook in $P$, say $D U^k D$, with $U^{k-1}DUD$ (resp. with $DUDU^{k-1}$),
  where $\varphi(P) = \psi(P) = P$ if $P$ contains no up hooks.
  See \fref{fig:78uijvkn} which shows the deformed paths $\varphi(P)$ and $\psi(P)$
  obtained from the positive grounded path $P$ in \fref{fig:hooks}.
  \begin{figure}
    \centering
    \includegraphics{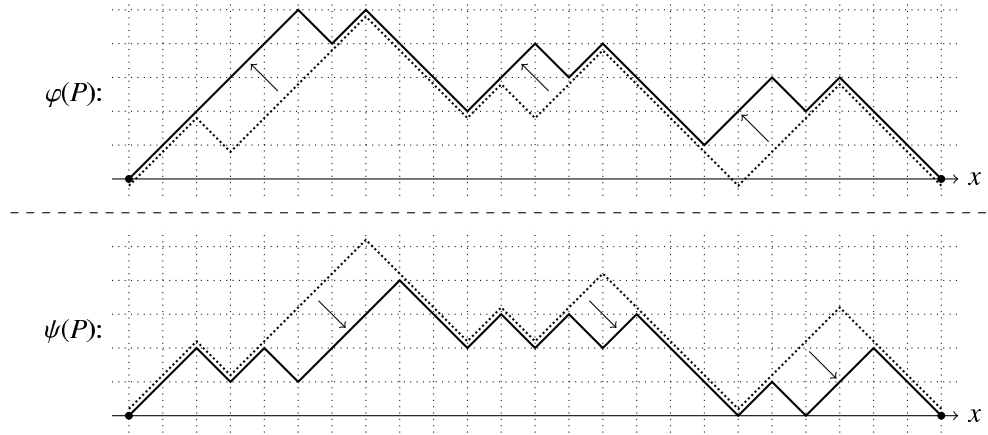}
    \caption{
      The deformation by the maps $\varphi$ and $\psi$.
      The paths $\varphi(P)$ and $\psi(P)$ (depicted in solid lines) are obtained
      from the path $P$ (in dotted lines) in \fref{fig:hooks}
      by deforming each up hook in $P$ as indicated by an arrow.
    }
    \label{fig:78uijvkn}
  \end{figure}
  
  We can observe the following on the maps $\varphi$ and $\psi$:
  Let $P$ and $P'$ be positive grounded paths.
  \begin{enumerate}[(a)]
  \item \label{itm:uya76dtuf}
    The $\varphi$ and $\psi$ never increase the number of up hooks in $P$
    as well as the number of down hooks in $P$.
  \item \label{itm:de8hckuvf}
    If $P$ contains no up hooks to the right side of the line $x=j$ then $\varphi(P)$ so to $x=j-2$.
    Similarly, if $P$ contains no up hooks to the left side of $x=j$ then $\psi(P)$ so to $x=j+2$.
  \item \label{itm:abukwywe}
    The $\varphi$ and $\psi$ never move the location of the two ends of $P$.
  \item \label{itm:sieufscds}
    If $P$ and $P'$ are non-intersecting then $\varphi(P)$ and $\varphi(P')$ are so.
    That is also the case with $\psi(P)$ and $\psi(P')$.
  \item \label{itm:hy7idusc}
    The following `mean' formula holds:
    \begin{equation} \label{eq:ukdhcis}
      W(\varphi(P)) + W(\psi(P)) = 2 W(P).
    \end{equation}
  \end{enumerate}

  Now, let us prove lemma \ref{lem:7gidlvbl}.
  Let $\bm{P} = {\{ P_0,\ldots,P_{n-1} \}} \in P(t,n)$.
  We can assume that $\bm{P}$ takes the minimal weight, $W(\bm{P}) = \min_{\bm{P}' \in P(t,n)} W(\bm{P}')$,
  without any loss of generality.
  For each $k \in \mathbb{N}_{0}$, let $\varphi^{k}(\bm{P}) = {\{ \varphi^k(P_0),\ldots,\varphi^k(P_{n-1}) \}}$
  denote the $n$-sets of positive grounded paths obtained from $\bm{P}$
  by applying the map $\varphi$ iteratively $k$ times to each path $P_j \in \bm{P}$.
  By induction with respect to $k \in \mathbb{N}_{0}$, we can show the following:
  \begin{enumerate}[(i)]
  \item
    Due to the observation (\ref{itm:de8hckuvf}),
    every path in $\varphi^k(\bm{P})$ contains no up hooks to the right side of the line $x=2t-2k-1$.
    That is because $\bm{P}$ contains no up hooks to the right side of $x=2t-1$.
  \item
    Due to the observations (\ref{itm:abukwywe}) and (\ref{itm:sieufscds}), $\varphi^k(\bm{P}) \in P(t,n)$.
  \item
    Due to the observation (\ref{itm:hy7idusc}), $W(\varphi^k(\bm{P})) = W(\bm{P})$.
    Indeed, if $W(\varphi^k(\bm{P})) > W(\varphi^{k-1}(\bm{P}))$, the formula \eref{eq:ukdhcis} would lead
    $W(\psi \circ \varphi^{k-1}(\bm{P})) < W(\varphi^{k-1}(\bm{P}))$,
    that contradicts the minimality of $W(\bm{P})$.
  \end{enumerate}
  As a consequence of (i), (ii), (iii), we can deduce that $\varphi^{t-1}(\bm{P}) \in P(t,n)$ contains no up hooks
  and has the minimal weight $W(\varphi^{t-1}(\bm{P})) = W(\bm{P})$.

  In a similar way, we can show the existence of $\check{\bm{P}} \in P(t,n)$ containing no down hooks
  and having the minimal weight $W(\check{\bm{P}}) = W(\bm{P})$.
  From the discussion in the last paragraph, the $\varphi^{t-1}(\check{\bm{P}})$ contains no up hooks
  and has the minimal weight $W(\varphi^{t-1}(\check{\bm{P}})) = W(\bm{P})$.
  Further, due to the observation (\ref{itm:uya76dtuf}), the $\varphi^{t-1}(\check{\bm{P}})$ contains no down hooks.
  Therefore, the $\varphi^{t-1}(\check{\bm{P}})$ having the minimal weight belongs to the set $\bar{P}(t,n)$.
  That completes the proof.
\end{proof}

As a consequence of lemma \ref{lem:7gidlvbl},
we can also solve the initial value problem of the ultradiscrete Toda molecule in the following way:

\begin{thm} \label{thm:f7iewchds}
  The solution to the initial value problem of the ultradiscrete Toda molecule \eref{eq:uToda} is given
  by \eref{eq:oje8wxas} with the tau function
  \begin{equation} \label{eq:qjf0ew9vn}
    T^{(t)}_{n} = \min_{\bar{\bm{P}} \in \bar{P}(t,n)} W(\bar{\bm{P}}), \qquad t,n \in \mathbb{N}_{0}.
  \end{equation}
\end{thm}

The expression \eref{eq:qjf0ew9vn} of the tau function $T^{(t)}_{n}$ is supported by the set $\bar{P}(t,n)$
much smaller than $P(t,n)$ in \eref{eq:nsiygivf}.
In fact, the cardinarity of $\bar{P}(t,n)$ is equal to the binomial number
\begin{equation}
  \#{\bar{P}(t,n)} = {t+n-1 \choose t} = \prod_{1 \le j < t} \frac{n+j}{j}
\end{equation}
which is much smaller than $\#{P(t,n)}$ given by \eref{eq:fweo8ger}.
In this sense, the expression \eref{eq:qjf0ew9vn} gives a simpler expression of the tau function $T^{(t)}_{n}$
than \eref{eq:nsiygivf}.

\subsection{Solution in shortest paths on a graph}
\label{sec:jenbrv98}

In \sref{sec:jenbrv98}, based on theorem \ref{thm:f7iewchds},
we derive another combinatorial expression of the solution in terms of {\em shortest paths} on a graph.

Let $G$ denote the (directed acyclic) graph in $\mathbb{N}_{0}^2$
consisting of the vertices at the points $(j,k) \in \mathbb{N}_{0}^2$, $j \ge k$,
connected by the two types of (directed) edges:
{\em east edges} $E_{j,k} = (2,0)$ and {\em south edges} $S_{j,k} = (0,-1)$.
Here the subscripts $j,k$ indicate that the initial points of $E_{j,k}$ and $S_{j,k}$ are at the vertex $(j,k)$.
As shown in \fref{fig:gadcsnvdh}, since the east edges $E_{j,k}$ have length two,
the graph $G$ splits into two disjoint subgraphs.
\begin{figure}
  \centering
  \includegraphics{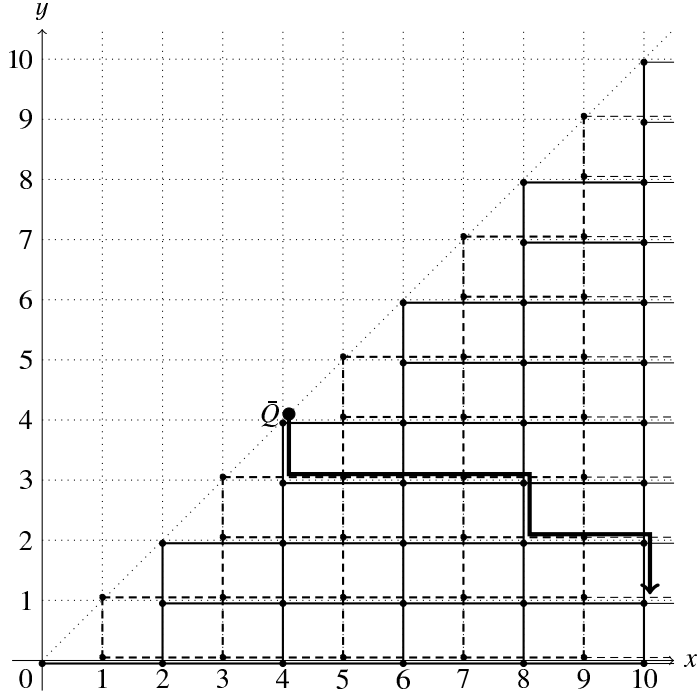}
  \caption{
    The graph $G$ and a path $\bar{Q} \in \bar{Q}(t,n)$, where $t=4$, $n=3$.
    The graph $G$ splits into two disjoint subgraphs,
    one of which is drawn in solid lines and the other in dashed lines.
    The path $\bar{Q}$, drawn in thick lines, goes between $(t,t)$ and $(t+2n,1)$.
  }
  \label{fig:gadcsnvdh}
\end{figure}

We define the weight function $W$ over the edges in $G$ by
\begin{equation} \label{eq:hcdihvy}
  W(E_{j,k}) = \sum_{\ell=0}^{j-k-1} A_{\ell} + k A_{j-k}, \qquad
  W(S_{j,k}) = 0,
\end{equation}
where $A_n$ are the initial value for the ultradiscrete Toda molecule.
We think of the weight $W(e)$ of an edge $e$ as the {\em length} of $e$.
For each path $Q$ on $G$, we then think of the weight $W(Q)$ as the {\em length} of $Q$,
which is equal to the sum of the weight $W(e)$ of all the edges $e$ passed by $Q$.
(The length $W(Q)$ may be negative due to the arbitrariness in $A_n$.)

For $t \in \mathbb{N}_{0}$, $t \ge 1$, and $n \in \mathbb{N}_{0}$,
let $\bar{Q}(t,n)$ denote the collection of paths on $G$ between the two vertices $(t,t)$ and $(t+2n,1)$.
We then have a one-to-one correspondence
between $\bar{\bm{P}} = {\{ \bar{P}_0,\ldots,\bar{P}_{n-1} \}} \in \bar{P}(t,n)$ and $\bar{Q} \in \bar{Q}(t,n)$:
The tabular positive grounded path $\bar{P}_j$ has its peaks and valleys
in the strip bordered by the lines $y=k-2j$ and $y=k-2j+1$
if and only if the path $\bar{Q}$ on $G$ passes through the east edge $E_{t+2j,t-k}$.
For example, in the $n$-sets $\bm{P} \in \bm{P}(t,n)$ of positive grounded paths in \fref{fig:iaufg8vyu},
the lower left one belongs to $\bar{\bm{P}}(t,n)$ and is in one-to-one correspondence
with the path $\bar{Q}$ in \fref{fig:gadcsnvdh}.
Actually, the weight function $W$ on $G$ is defined in \eref{eq:hcdihvy}
so that $W(\bar{\bm{P}}) = W(\bar{Q})$ for every pair of $\bar{\bm{P}}$ and $\bar{Q}$
in one-to-one correspondence.
We thereby have the identity
\begin{equation} \label{eq:uyaivsd}
  \min_{\bar{Q} \in \bar{Q}(t,n)} W(\bar{Q}) = \min_{\bar{\bm{P}} \in \bar{P}(t,n)} W(\bar{\bm{P}}).
\end{equation}

For $t,n \in \mathbb{N}_{0}$, let us define $Q(t,n)$ to be the collection of paths on $G$
between the two vertices $(t,t)$ and $(t+2n,0)$.
We can then show that the identity \eref{eq:uyaivsd} still holds
even if we replace the support set $\bar{Q}(t,n)$ of the left-hand minimum with $Q(t,n)$,
\begin{equation} \label{eq:spj0g8bre}
  \min_{Q \in Q(t,n)} W(Q) = \min_{\bar{\bm{P}} \in \bar{P}(t,n)} W(\bar{\bm{P}}),
\end{equation}
for $W(E_{j,0}) = W(E_{j,1})$ for every $j \in \mathbb{N}_{0}$, $j \ge 1$.
In addition, in that case, the identity \eref{eq:spj0g8bre} also takes place for $t=0$.

Finally, combining theorem \ref{thm:f7iewchds} and the identity \eref{eq:spj0g8bre}, we obtain the following result:

\begin{thm} \label{thm:uTodaSol}
  The solution to the initial value problem of the ultradiscrete Toda molecule \eref{eq:uToda} is given
  by \eref{eq:oje8wxas} with the tau function
  \begin{equation} \label{eq:dahcbfjvs}
    T^{(t)}_{n} = \min_{Q \in Q(t,n)} W(Q), \qquad t,n \in \mathbb{N}_{0},
  \end{equation}
  where the weight function $W$ on the graph $G$ is defined by \eref{eq:hcdihvy}.
\end{thm}

The right-hand side of \eref{eq:dahcbfjvs} denotes
the length of the {\em shortest paths} on the graph $G$ between the two vertices $(t,t)$ and $(t+2n,0)$.
It should be noted that Nakata \cite{Nakata(2011)} constructed a similar combinatorial expression
in terms of shortest paths (called {\em minimum weight flows} in \cite{Nakata(2011)})
of a particular solution to the ultradiscrete Toda molecule on the finite lattice, $n = 0,1,\ldots,N$.

\section{Concluding remarks}
\label{sec:ckdoc89p}

In this paper, we have investigated the discrete and ultradiscrete Toda molecules
from a combinatorial viewpoint.
To the tau function which solves an initial value problem of the discrete Toda molecule,
we have given a combinatorial expression in terms of non-intersecting paths.
Especially, in order to read the tau function in combinatorial words,
we utilized Flajolet's path interpretation of continued fractions
and Gessel--Viennot's lemma on determinants and non-intersecting paths.
Due to the combinatorial expression, we have succeeded to derive a subtraction-free expression of
the tau function which is given in a Hankel determinant.

For an initial value problem of the ultradiscrete Toda molecule,
we first obtained an exact solution by ultradiscretizing the tau function of the discrete Toda molecule.
We next rewrote the tau function to derive a simpler expression in terms of shortest paths.
As a result, we have shown that the tau function
which solves the initial value problem of the ultradiscrete Toda molecule
can be evaluated as the length of shortest paths on a specific graph in which the length of edges is
determined by the initial value.

In this paper, we deduced combinatorial expressions of the tau functions
with the help of a determinant solution to the discrete Toda molecule and the technique of ultradiscretization.
For the tau functions \eref{eq:iu8o7yew} and \eref{eq:dahcbfjvs} given in terms of combinatorial objects,
however, it is expected to make combinatorial (or bijective) proofs to directly verify that
the tau functions satisfy the bilinear equations \eref{eq:suhifefvs} and \eref{eq:ajekuygw}.
For that purpose, the technique of alternating walks \cite{Goulden(1988)}
for Schur symmetric functions would be useful.

The combinatorial idea used in this paper should be applicable to other discrete integrable systems
associated with continued fractions,
such as the $R_{\mathrm{I}}$ and $R_{\mathrm{II}}$ chains with $R_{\mathrm{I}}$- and $R_{\mathrm{II}}$-fractions
\cite{Vinet-Zhedanov(1998),Spiridonov-Zhedanov(1999)},
the FST chain with the Thiele-type continued fractions \cite{Spiridonov-Tsujimoto-Zhedanov(2007)},
and the matrix qd algorithm with the matrix S-fractions \cite{VanIseghem(2004)}.
Those applications will be discussed in future works.

\section*{Acknowledgments}

The authors would like to thank Professor Yoshimasa Nakamura for valuable discussions and comments.
This work was supported by JSPS KAKENHI Grant Number 24740059.


\end{document}